\documentclass[letter,11pt]{article}
\usepackage[letterpaper,text={6.6in,9.3in}]{geometry}

\usepackage{setspace}
\usepackage{authblk}
\usepackage{hyperref}
\usepackage{amsfonts}
\usepackage{amsmath}
\usepackage{amssymb}
\usepackage{amsthm}
\usepackage{graphicx}
\usepackage{color}
\usepackage[none]{hyphenat}
\usepackage{wrapfig}
\usepackage[justification=centering]{caption}
\usepackage{bm}
\usepackage{enumitem}

\theoremstyle{plain}
\newtheorem{theorem}{Theorem}[section]
\newtheorem{lemma}[theorem]{Lemma}
\newtheorem{claim}{Claim}
\newtheorem{corollary}[theorem]{Corollary}
\theoremstyle{definition}
\newtheorem{definition}[theorem]{Definition}

\usepackage{algorithm}
\usepackage[noend]{algpseudocode}

\theoremstyle{plain}
\theoremstyle{definition}

\graphicspath{ {Figures/} }

\DeclareMathOperator*{\argmax}{arg\,max}

\newcommand{\dx}{\mathrm{d}x}

\begin{document}
	
\pagenumbering{gobble}
\begin{titlepage}
{ 
	\title{Reaching Distributed Equilibrium with Limited ID Space
		{\let\thefootnote\relax\footnote{
			This research was supported by the Israel Science Foundation (grant 1386/11).
		}}
	}
	\author{Dor Bank}
	\author{Moshe Sulamy}
	\author{Eyal Waserman}
	\affil{Tel-Aviv University}
	\date{}
	\maketitle
}

\begin{abstract}
	
We examine the relation between the size of the $id$ space
and the number of rational agents in a network
under which equilibrium in distributed algorithms is possible.
When the number of agents in the network is \emph{not} a-priori known,
a single agent may duplicate to gain an advantage, pretending to be more than one agent.
However, when the $id$ space is limited, each duplication involves a risk of being caught.
By comparing the risk against the advantage,
given an $id$ space of size $L$,
we provide a method of calculating the minimal threshold $t$,
the required number of agents in the network,
such that the algorithm is in equilibrium.
That is, it is the minimal value of $t$ such that if agents a-priori know that $n \geq t$
then the algorithm is in equilibrium.
We demonstrate this method by applying it to two problems, Leader Election and Knowledge Sharing,
as well as providing a constant-time approximation $t \approx \frac{L}{5}$
of the minimal threshold for Leader Election.


\end{abstract}

	
\end{titlepage}
\pagenumbering{arabic}

\section{Introduction}
\label{section_intro}

Traditionally, distributed computing focuses on computing in the face of faulty processors,
from fail-stop to Byzantine faults, and more.
Recently, a new model of distributed game theory has emerged,
in which the participants, now called rational agents,
are not faulty but may cheat in order to increase their profit.
The goal is to design distributed algorithms that are \emph{in equilibrium},
that is, where no agent has an incentive to cheat.

This paper addresses the impact of the $id$ space on the equilibrium
of distributed algorithms where the network consists of rational agents.
When the $id$ space is limited and the number of agents in the network is not a-priori known,
what is the minimal guarantee on the number of agents (the \emph{threshold})
we must provide agents for the algorithm to reach equilibrium?

To address this question, we consider a network of \emph{rational} agents \cite{Abraham:2011}
with a limited $id$ space of size $L$.
Agents a-priori know $L$ but do not know $n$, the number of agents in the network.
Additionally, we provide agents with a threshold $t$,
the minimal number of agents in the network such that $t \leq n$,
i.e., agents a-priori know that the number of agents $n$ is distributed uniformly $n\sim U[t,L]$.
Our goal is to find the \emph{minimal} value of $t$ for an algorithm
such that it is in equilibrium.

Stemming from game theory, the agents in the network are \emph{rational}.
Each agent $a$ has a utility function $u_a$ assigning values to all possible outcomes of the algorithm;
the higher the value, the better the outcome is for $a$.
A rational agent participates in the algorithm but may deviate from it 
if a deviation increases the probability of its preferred outcome,
i.e., its \emph{expected} utility increases.
An algorithm is said to be \emph{in equilibrium} if no agent has an incentive to deviate
at any point in the algorithm execution.
To differentiate rational agents from Byzantine faults,
the \emph{Solution Preference} \cite{gtblocks} property ensures
agents never prefer an outcome in which the algorithm fails
over an outcome in which the algorithm terminates correctly.

Here we study the case where $n$, the number of agents in the network,
is not a-priori known to agents.
In such a case, an agent can duplicate itself in order to gain an advantage,
deviating from the algorithm by pretending to be more than one agent.
This deviation is also known as a Sybil Attack \cite{sybilattack}.
For example, in a \emph{fair} Leader Election protocol,
where every agent has an equal chance of being elected leader,
a deviating agent can duplicate itself to gain an \emph{unfair} advantage
and increase its chances of being elected leader (itself or one of its simulated duplicates).

In \cite{gtcolor} the authors study the case where agents do not a-priori know $n$,
or know only \emph{bounds} $\alpha,\beta$ on the value of $n$ such that $\alpha \leq n \leq \beta$.
In both cases it is assumed the $id$ space is much larger than $n$,
disregarding the probability that a fake $id$ (of a duplicated agent) collides with an existing $id$,
thus an agent can duplicate itself with low or no risk involved.

In this paper we examine the case where the $id$ space \emph{is} limited,
such that \emph{any} duplication involves a greater risk.
Agents a-priori know only $L$ and the threshold $t$, that is,
the minimal number of agents in the network such that $t \leq n \leq L$.
Our goal is to find the \emph{minimal} threshold $t$ as a function of $L$
for an algorithm to be in equilibrium.

Our contributions are as follows:
\begin{itemize}
	\item A general method for calculating the minimal threshold $t$
	for \emph{any} algorithm in $O(L^3)$ time.
	
	\item Two enhancements to the general method.
		The "Linear Threshold" enhancement allows us to use binary search for the value of $t$,
		instead of scanning all possible values,
		thus improving the running time to $O(L^2 \log L)$.
		The "Limited Duplications" enhancement allows us to look at a single duplication number
		instead of checking all possible values of $m$,
		thus improving the running time to $O(L^2)$.
	
	\item Leader Election algorithm \cite{gtleader,gtblocks} is proven to satisfy both enhancements,
		thus finding the threshold in $O(L \log L)$ time.
		Particularly, we show that whenever an agent has an incentive to deviate by duplicating
		$m$ agents, it also has an incentive to deviate by duplicating $1$ agent.
		Thus, to check for equilibrium it suffices to check the case $m=1$.
	
	\item Constant-time approximation of the Leader Election threshold,
		such that for a threshold $t \approx \frac{L}{5}$ the algorithm is in equilibrium.
	
	\item Knowledge Sharing in a ring \cite{gtcolor} is proven to satisfy the "Linear Threshold"
	enhancement (Section~\ref{sub_enhance}), thus finding the threshold in $O(L^2 \log L)$ time.
\end{itemize}

\subsection{Related Work}

Recently, there has been a line of research studying the connection between
distributed computing and game theory,
stemming from the problem of secret sharing \cite{secretsharing}
and multiparty computation \cite{Abraham:2011,scalablerational,byzagreement},
as well as the BAR model (Byzantine, acquiescent and rational) \cite{barfault,fairsfe,selfishevil}.
In \cite{AbrahamDGH06,Abraham:2011}, the authors discuss the combination of these two fields.
In \cite{gtleader}, the first protocols where processors in the network behave as rational agents
are presented, specifically protocols for Leader Election.
This line of research is continued in \cite{gtblocks}, where the authors provide
basic building blocks for game theoretic distributed networks,
and algorithms for Consensus, Renaming, Leader Election, and Knowledge Sharing.
Consensus was researched further in \cite{gtconsensus},
where the authors show that there is no ex-post Nash equilibrium, and present
a Nash equilibrium that tolerates $f$ failures under some minimal assumptions on the failure pattern.

In \cite{gtcolor}, the authors examine the impact of a-priori knowledge of the network size
on the equilibrium of distributed algorithms.
When agents have no a-priori knowledge of $n$, or only know bounds on $n$,
a cheating agent may duplicate itself to appear as if it were several agents in the network,
and thus may gain an advantage in certain algorithms.
In \cite{gtcolor}, the authors assume the $id$ space is unlimited,
i.e., some duplications involve minimal or no risk.
To the best of our knowledge,
distributed algorithms with rational agents where the $id$ space is limited
have not been studied previously.

\section{Model}
\label{section_model}
The model in discussion is the standard synchronous message-passing model
with $n \geq 3$ nodes, each node representing an agent.
The network is assumed to be $2$-vertex connected, as was shown necessary in \cite{gtblocks}.
Throughout the paper, $n$ always denotes the actual number of nodes in the network.

Initially, each agent knows its $id$ and input (if any), but not the $id$ and input of any other agent.
We assume the prior of each agent over any information it does not know is uniformly distributed
over all possible values.
The size of the $id$ space is $L$, such that each agent is assigned a unique $id$ taken from $[1,L]$.
The threshold $t$ is a minimal guarantee on the size of the network such that $n \geq t \geq 3$.
Both $L$ and $t$ are a-priori known to all agents.

We assume all agents start the protocol together.
If not, we can use the Wake-Up building block \cite{gtblocks} to relax this assumption.

\subsection{Equilibrium}

An algorithm is in \emph{equilibrium} if no agent has an incentive to deviate.
When an agent considers a deviation, it assumes all other agents follow the algorithm,
i.e., it is the only agent deviating.

Following the model in \cite{gtleader,gtblocks,gtcolor} we assume an agent aborts the algorithm
whenever it detects a deviation made by another agent\footnotemark,
even if the detecting agent can gain by not aborting.

\footnotetext{
	Such as receiving the wrong message, learning of two agents with the same $id$, etc.
}

Each algorithm produces a single output per agent, once, at the end of the execution.
Each agent is a \emph{rational} agent, and has a preference over its own output.
Formally, let $\Theta$ be the set of all possible output vectors,
denote an output vector $O=\{o_1,\dots,o_n\} \in \Theta$,
where $o_a$ is the output of agent $a$ and $O[a]=o_a$.
Let $\Theta_L$ be the set of \emph{legal} output vectors, in which the protocol terminates successfully,
and let $\Theta_E$ be the set of \emph{erroneous} output vectors
such that $\Theta = \Theta_L \cup \Theta_E$ and $\Theta_L \cap \Theta_E = \varnothing$.

Each rational agent $a$ has a utility function $u_a: \Theta \rightarrow N$.
The higher the value assigned by $u_a$, the better the output is for $a$.
We assume all utility functions satisfy \emph{Solution Preference}
\cite{gtleader,gtblocks} which guarantees that an agent never has an incentive to cause the algorithm to fail.

\begin{definition}[Solution Preference]
	The utility function $u_a$ of an agent $a$ never assigns a higher utility to an erroneous output
	than to a legal output:
	$$\forall{a,O_L\in\Theta_L,O_E\in\Theta_E}: u_a(O_L) \geq u_a(O_E)$$
\end{definition}

We differentiate \emph{legal} output vectors, which ensure the output is valid and not erroneous\footnotemark,
from the \emph{correct} output vectors, which are output vectors that are a result of a correct
execution of the algorithm, i.e., without any deviation.
The solution preference guarantees agents never prefer an erroneous output over a legal one;
however, they may prefer a \emph{legal} but \emph{incorrect} output.

\footnotetext{
	For example in a consensus algorithm,
	an output of  $1$ when all agents received $0$ as input is \emph{erroneous}.
	An output of $1$ in an algorithm deciding according to the majority,
	where some agents received $1$ as input but the majority received $0$,
	is still \emph{legal} but also \emph{incorrect}.
}

For simplicity, we assume agents only have preferences over their own output,
and each agent $a$ prefers a \emph{single} value $p_a$.
We normalize the utility function for any agent $a$ to Equation~\ref{eq_norm};
however, our results hold for \emph{any} utility function that satisfies Solution Preference.

\begin{equation}
\label{eq_norm}
u_a(O)=\begin{cases}
1 & o_a=p_a$ and $O \in \Theta_L
\\ 0 & o_a \neq p_a$ or $O \in \Theta_E\end{cases}
\end{equation}

At each round in the algorithm execution, each agent $a$ chooses which step to take for the following round,
where a step consists of all actions for that round,
e.g., drawing a random number, performing internal computation, the contents of any sent messages.
If a step exists that does not follow the algorithm and improves its \emph{expected} utility,
we then say the agent has an \emph{incentive} to deviate.

\begin{definition}[Expected Utility]
	For each possible output vector $O$, let $x_O(s)$ be the probability
	estimated by agent $a$ that $O$ is output by the algorithm if $a$ takes step $s$
	and all other agents follow the algorithm.
	The \emph{expected utility} of agent $a$ for step $s$ is:
	$$\mathbb{E}_{s}[u_a] = \sum\limits_{o\in\Theta} x_O(s) u_a(O)$$
\end{definition}

If no agent has an incentive to deviate from the algorithm,
when assuming all other agents follow the algorithm, we say that the algorithm is in equilibrium.
In this paper we assume a single deviating agent, i.e., there are no coalitions of agents.

\begin{definition}[Distributed Equilibrium]
	Let $\Lambda(r)$ denote the next step of algorithm $\Lambda$ in round $r$.
	$\Lambda$ is in equilibrium if for any step $s$, at any round $r$
	of every possible execution of $\Lambda$:
	$$\forall{a,r,s}: \mathbb{E}_{\Lambda(r)}[u_a] \geq \mathbb{E}_s[u_a] $$
\end{definition}

\subsection{Duplication}
\label{model_dup}

When $n$ is not a-priori known to agents,
an agent $a$ can deviate by simulating $m$ imaginary agents.
The size of the $id$ space is $L$,
thus each duplicated agent must have a unique $id$ and duplication involves a risk
of choosing an $id$ that already exists, causing the algorithm failure.

We assume all algorithms force a deviating agent to commit to its duplication early,
by starting the algorithm with a process that tries to learn all $id$s,
such as the Wake-Up \cite{gtblocks} algorithm\footnotemark.
\footnotetext{
	We can add the Wake-Up building block at the beginning of any algorithm,
	otherwise the deviating agent can only gain by duplicating itself at a later stage.
}

The number of duplications of a deviating agent (not including itself)
is denoted by $m$.
We assume $m$ and the $id$s of all $m$ duplicated agents are chosen at round $0$,
before the algorithm starts.
The case where $m$ and $id$s are chosen during wake-up
is left as a direction for future research (Section~\ref{section_discussion}).

Denote by $e_0(x)$ the expected utility of an agent at round $0$ if it follows the protocol
and the network consists of $x$ agents.
The value of $e_0(x)$ is identical for all agents,
else the algorithm is not in equilibrium\footnotemark.
Denote by $e_m(x)$ the expected utility at round $0$ if the agent deviates by duplicating $m$ agents
\emph{successfully}, i.e., disregarding the risk of duplicated $id$s.
\footnotetext{
	An agent's unique knowledge at the beginning of the algorithm is its $id$ and input,
	thus if different expected utilities exist, an agent has an incentive to deviate
	from the algorithm and pretend it has a different $id$ or input value.
}

\subsection{Leader Election}

In the Leader Election problem, each agent $a$ outputs $o_a \in \{0,1\}$,
where $o_a=1$ means that $a$ was elected leader, and $o_a=0$ means otherwise.
The set of legal output vectors is defined as:
$$O_L = \{O | \exists{a}:o_a=1, \forall{b\neq a}: o_b=0 \}$$

We assume a \emph{fair} leader election\cite{gtleader} where, at the beginning of the algorithm,
each agent has an equal chance to be elected leader, i.e., $e_0(i) = \frac{1}{i}$.
We assume agents prefer $1$, i.e., to be elected leader.

\subsection{Knowledge Sharing}

The Knowledge Sharing problem (from \cite{gtcolor}, adapted from \cite{gtblocks}) is defined as follows:

\begin{enumerate}
	
	\item Each agent $a$ has a private input $i_a$ and a function $q$,
	where $q$ is identical at all agents.
	
	\item An output is \emph{legal} if all agents output the same value, i.e.,
	$O \in \Theta_L \Longleftrightarrow \forall{a,b: O[a] = O[b]}$.
	
	\item An output is \emph{correct} if all agents output $q(I)$ where $I=\{i_1,\dots,i_n\}$.
	
	\item The function $q$ satisfies the Full Knowledge property\cite{gtblocks,gtcolor}:
	\begin{definition}[Full Knowledge]
		\label{def_fullknow}
		When one or more input values are not know,
		any output in the range of $q$ is \emph{equally} possible.
		Formally, for any $1 \leq j \leq k$ denote $z_y=|\{x_j|q(i_1,\dots,x_j,\dots,i_m)=y\}|$.
		For any possible output $y$ in the range of $q$, $z_y$ is the same\footnotemark.
		\footnotetext{
			Assuming input values are drawn uniformly, otherwise the definition of $z_y$
			can be expanded to the sum of probabilities over every input value.
		}
	\end{definition}
\end{enumerate}

We assume that each agent $a$ prefers a certain output value $p_a$.
Following \cite{gtcolor}, in this paper we only discuss Knowledge Sharing in ring graphs.

\section{Solution Basis}
\label{section_basis}


The threshold $t$ is a-priori known to all agents and satisfies $t \leq n$.
Our goal is to find the minimal $t$ (as a function of $L$) that satisfies Equation~\ref{eq_basis},
i.e., the minimal threshold for which the algorithm is in equilibrium for any number of duplications $m$.
Here we show a method that constructs Equation~\ref{eq_basis}.

\begin{equation}
\label{eq_basis}
\sum\limits_{x=t}^L e_0(x) \geq \max_m \sum\limits_{x=t}^{L-m} p_m(x) e_m(x)
\end{equation}

Calculating the minimal threshold $t$ from Equation~\ref{eq_basis} can be done in $O(L^3)$ running time.
In Section~\ref{sub_enhance} we describe two enhancements that improve this running time.

\subsection{Constructing Equation \ref{eq_basis}}

Recall $e_0(x)$ is the expected utility if following the algorithm and $n=x$,
and $e_m(x)$ is the expected utility if one agent duplicates itself $m$ times \emph{successfully}.
When an agent $a$ does not know the actual number of agents in the network it assumes $n \sim U[t,L]$.
The expected utility of $a$, considering all possible values of $n$
for a threshold $t$ and $id$ space of size $L$, is $f(L,t)$ given in Equation~\ref{eq_no_cheat}:
\begin{equation}
\label{eq_no_cheat}
f(L,t) = \frac{1}{L-t+1} \sum\limits_{x=t}^L e_0(x)
\end{equation}

When the number of duplication $m$ satisfies $m+n>L$,
agent $a$ can easily be caught since counting the number of agents in the network
will exceed $L$ (and an $id$ collision exists), thus its expected utility is $0$,
i.e., $\forall{x>(L-m)}: e_m(x) = 0$.

When $m+n \leq L$,
agent $a$ can be caught if one or more of the chosen $id$s for its duplicates
already exists in the network.
Agent $a$ knows only its own $id$ at the beginning of the algorithm and
since $id$s are chosen in advance (see model, Section~\ref{section_model}),
$a$ needs to choose $m$ $id$s out of $L-1$ possibilities,
of which only $L-n$ are available, $id$s that do not already exist in the network.
Denote $p_m(x)$ the probability of choosing $m$ $id$s that do not collide with any $id$
in a network of $x$ agents, such that: $p_m(x) = \frac{\binom{L-x}{m}}{\binom{L-1}{m}} $.
The expected utility of agent $a$ for $1 \leq m \leq (L-x)$ is thus:
$p_m(x) e_m(x)$.

Finally, considering all possible values of $n$,
the expected utility of agent $a$ for a threshold $t$ and $id$ space of size $L$ is $g(L,t,m)$
given in Equation~\ref{eq_cheat}.
An agent thus chooses $m=\argmax_m g(L,t,m)$.
\begin{equation}
\label{eq_cheat}
g(L,t,m) = \frac{1}{L-t+1} \sum\limits_{x=t}^{L-m} p_m(x) e_m(x)
\end{equation}

When agents a-priori know $t$ and $L$,
the algorithm is in equilibrium when $f(L,t) \geq \max_m g(L,t,m)$,
i.e., $\argmax_m g(L,t,m) = 0$, since by definition $f(L,t) = g(L,t,0)$.
This holds when Equation~\ref{eq_basis_complex} is satisfied,
from which we can easily deduce Equation~\ref{eq_basis}.
\begin{equation}
\label{eq_basis_complex}
\begin{aligned}
\frac{1}{L-t+1} \sum\limits_{x=t}^L e_0(x) \geq
\max_m \frac{1}{L-t+1} \sum\limits_{x=t}^{L-m} p_m(x) e_m(x)
\end{aligned}
\end{equation}

\subsection{Enhancements}
\label{sub_enhance}

The minimal threshold $t$ can be found in $O(L^3)$ time,
by going through all possible values of $t$ and $m$
($L$ possible values for each, and each computation takes $O(L)$ time).

Here we show two enhancements that can be applied for certain algorithms to improve the running time.
The first enhancement improves the running time to $O(L^2\log L)$ and applies to most algorithms,
and the second enhancement improves it to $O(L^2)$ and only applies to specific algorithms.
Applying both enhancements improves the running time to $O(L \log L)$.


\subsubsection{Linear Threshold}
\label{enhance1}

For most algorithms, some $L_0$ exists such that for any $L>L_0$
a specific pivot value $t_0$ exists,
such that for any $t \geq t_0$ the algorithm is in equilibrium,
and for any $t < t_0$ it is not in equilibrium.

Following equations \ref{eq_no_cheat} and \ref{eq_cheat},
this applies to algorithms which satisfy Equation~\ref{eq_linear}:
\begin{equation}
\label{eq_linear}
\exists{L_0}\forall{L>L_0} \exists{t_0} \forall{t}:
\begin{cases}
	t \geq t_0 \Longleftrightarrow f(L,t) \geq \max_m g(L,t,m) \\
	t < t_0 \Longleftrightarrow f(L,t) < \max_m g(L,t,m)
\end{cases}
\end{equation}

For algorithms that satisfy Equation~\ref{eq_linear},
we can use binary search to find the minimal threshold that satisfies Equation~\ref{eq_basis}.
Thus, we improve the running time to $O(L^2 \log L)$.

\subsubsection{Limited Duplications}
\label{enhance2}

For some algorithms it can be enough to look at a specific number of duplications $m'$,
such that if $m$ exists for which an agent $a$ has an incentive to deviate,
then $a$ also has an incentive to deviate with $m'$ duplications,
i.e., these algorithms satisfy Equation~\ref{eq_limited} \emph{for any $L$ and $t$}.

\begin{equation}
\label{eq_limited}
\exists{m}: g(L,t,m) > f(L,t) \Longleftrightarrow g(L,t,m') > f(L,t)
\end{equation}

For algorithms that satisfy Equation~\ref{eq_limited} for any $L$ and $t$
we only need to examine a single duplication value $m'$.
Thus, we improve the running time by a factor of $L$ to $O(L^2)$.

For an algorithm that satisfies Linear Threshold,
the running time is further improved to $O(L \log L)$.

\section{Leader Election}
\label{section_le}

Here we apply the method from Section~\ref{section_basis} to the problem of Leader Election,
and show that it satisfies both the Linear Threshold and Limited Duplications enhancements (Section~\ref{sub_enhance}),
thus the minimal threshold $t$ can be found in $O(L^2 \log L)$ time.
Furthermore, we show a constant-time approximation of $t\approx\frac{L}{5}$
for which Leader Election is in equilibrium, and the threshold is close to minimal.

\subsection{Equilibrium}
\label{le1}
Recall that $e_0(x)$ denotes the expected utility of an honest agent at the beginning of the algorithm, assuming there are $x$ agents in the network,
and $e_m(x)$ denotes the expected utility of an agent that deviates from the algorithm and duplicates itself $m$ times \emph{successfully}. 
In a fair Leader Election algorithm $e_0(x)=\frac{1}{x}$
and $e_m(x)=\frac{1+m}{x+m}$,
since the cheater emulates $m$ agents in addition to itself,
each having an equal chance of being elected leader from a total of $x+m$ agents 
(including duplications).

The probability of successfully duplicating $m$ times is $p_m(x)= \frac{\binom{L-x}{m}}{\binom{L-1}{m}}$.
Following Equation~\ref{eq_basis} from our solution basis (Section~\ref{section_basis}), the threshold to guarantee equilibrium is such $t$ that satisfies:
\begin{equation}
\label{eq_le_equil}
\sum\limits_{x=t}^L \frac{1}{x} \geq \max_m \sum\limits_{x=t}^{L-m} \frac{\binom{L-x}{m}}{\binom{L-1}{m}} \frac{1+m}{x+m}
\end{equation}

\subsection{Limited Duplication Enhancement}
\label{le2}
Here we show that in order to check equilibrium for a given $L,t$
it suffices to check a single duplication, i.e., $m=1$.
Simply put, it means that $\exists m : g(L,t,m) > f(L,t) \Longleftrightarrow g(L,t,1) > f(L,t)$.

\begin{theorem}
	Leader Election satisfies the requirements for Limited Duplications (Section~\ref{enhance2})
\end{theorem}

\begin{proof}
Right-to-left comes automatically, since $m'=1$ is such $m$ for which $g(L,t,m) > f(L,t)$ so existence holds. For the left-to-right claim, let us notice that for a fixed $n$, $\binom{n}{k}$ behaves like $O(\frac{n^k}{k!})$ as $k$ increases. This shows us that for a fixed value of $x$, the element $\frac{\binom{L-x}{m}}{\binom{L-1}{m}}$ decreases like $(\frac{L-x}{L-1})^m$ as a function of $m$. Meanwhile, for a fixed value of $x$, the element $\frac{1+m}{x+m}$ shows only bounded increase as a function of $m$, and is much less dominant than the former element for larger values of $m$. It becomes clear that either this function starts decreasing right from $m=1$, or that there is one global maximum $m^*$, at which the effect of $\frac{\binom{L-x}{m}}{\binom{L-1}{m}}$ becomes more dominant than the effect of $\frac{L-x}{L-1}$. So when $m^* > 1$, we get that $g(L,t,m^*) > g(L,t,1) > f(L,t)$. 
\end{proof}

\subsection{Linear Threshold Enhancement}
\label{le3}

\begin{theorem}
	Leader Election satisfies the requirements for Linear Threshold (Section~\ref{enhance1})
\end{theorem}

\begin{proof}
	Following the previous section, it is enough to prove the Linear Threshold requirement for $m=1$.
	When Equation~\ref{eq_linear_le} is satisfied, so is Equation~\ref{eq_linear}.
	Thus, we prove the following:
	\begin{equation}
	\label{eq_linear_le}
	\forall{L} \exists{t_0}:
	\begin{cases}
	t \geq t_0 \Longleftrightarrow f(L,t) \geq g(L,t,1) \\
	t < t_0 \Longleftrightarrow f(L,t) < g(L,t,1)
	\end{cases}
	\end{equation}
	
	
	For $t=L$: $f(L,L)=\frac{1}{L}$ and it holds that $n=t=L$,
	thus the number of agents in the network is a-priori known to all agents.
	By our assumption that the algorithm is in equilibrium when $n$ is known
	$\forall{L}: g(L,L,1)=0 < f(L,L)$.
	
	Fix $L$ and assume $f(L,3) < g(L,3,1)$. We will prove the inverse as well.

	By the intermediate value theorem there exists $t_0$ such that
	$f(L,t_0-1) < g(L,t_0-1,1)$ and $f(L,t_0) \geq g(L,t_0,1)$.
	We will prove that $\forall{t\geq t_0}: f(L,t) \geq g(L,t,1)$ by induction.
	
	Base: by definition $f(L,t_0) \geq g(L,t_0,1)$.
	Step: assume the claim holds for $t$, we will prove for $t+1$.
	
	\begin{claim}
	\label{le_lt_claim1}
		For every $3 \leq t \leq \frac{L-1}{2}$:
		$ \frac{1}{t} < \frac{2}{L-1} \frac{L-t}{t+1} $
	\end{claim}
	\begin{proof}
		$$\frac{1}{t} - \frac{2}{L-1} \frac{L-t}{t+1}=
		\frac{(L-1)(t+1)-2t(L-t)}{t(L-1)(t+1)} =
		\frac{-Lt+L-t-1+2t^2}{t(L-1)(t+1)} $$
		
		We wish to find where this expression is negative.
		The denominator is obviously positive so we can disregard it.
		
		$$ -Lt+L-t-1+2t^2 = 2t^2-t(1+L)+L-1=0 $$
		$$ t = \frac{1+L \pm \sqrt{(1+L)^2 - 8(L-1)}}{4} =
			\frac{1+L \pm (L-3)}{4}=1,\frac{L-1}{2} $$
		
		The nominator is convex so at the mid-range, the expression is negative.
		We can disregard $1$ since $t \geq 3$.
	\end{proof}

	For $t \geq \frac{L+1}{2}$, the risk of failing the algorithm satisfies
	$p_m(x) \leq \frac{L-(L+1)/2}{L-1}=\frac{1}{2}$.
	Thus it holds that $g(L,t,1) < f(L,t)$:
	$$g(L,t,1) \leq \frac{1}{L-t+1} \sum\limits_{x=t}^{L-1} \frac{1}{2} \frac{2}{x+1} =
		\frac{1}{L-t+1} \sum\limits_{x=t}^{L-1} \frac{1}{x+1} <
		\frac{1}{L-t+1} \sum\limits_{x=t}^{L} \frac{1}{x} = f(L,t) $$
	By the same assignment, it also holds for $t\geq\frac{L}{2}$.
	By the induction hypothesis and Claim~\ref{le_lt_claim1} we get:
	$$ \sum\limits_{x=t}^L \frac{1}{x} > \frac{2}{L-1} \sum\limits_{x=t}^{L-1} \frac{L-x}{x+1} $$
	$$ \frac{1}{t} + \sum\limits_{x=t+1}^L \frac{1}{x} > \frac{2}{L-1} \frac{L-t}{t+1}
		+ \frac{2}{L-1} \sum\limits_{x=t}^{L-1} \frac{L-x}{x+1} $$
	$$ \sum\limits_{x=t+1}^L \frac{1}{x} > \frac{2}{L-1} \sum\limits_{x=t+1}^{L-1} \frac{L-x}{x+1} $$
	The first transition is by definition,
	and the second is derived from Claim~\ref{le_lt_claim1}.
	Thus, we proved the induction and satisfy Equation~\ref{eq_linear_le},
	and following the same steps it is also satisfied when $f(L,3) \geq g(L,3,1)$.
\end{proof}

\subsection{Threshold Bounds}
\label{le4}
	
\begin{lemma}
	In the Leader-Election problem, for a sufficiently large $L$,
	a rational agent will cheat if $t\le0.2L$, but will not cheat if $t\ge0.21L$.
\end{lemma}

\begin{corollary}
	Leader Election is in equilibrium when all agents a-priori know $n \ge 0.21 L$
\end{corollary}

\begin{proof}
Since the Limited Duplication Enhancement applies to the Leader-Election problem, the condition for equilibrium at Equation~\ref{eq_le_equil} reduces to the following:
	\begin{equation}
	\label{eq_le_equil1}
	\sum\limits_{x=t}^L \frac{1}{x} \geq \frac{2}{L-1} \sum\limits_{x=t}^{L} \frac{L-x}{x+1} 
	\end{equation}

In the following sections we will show these greater and lower bounds:
$$ f_{lb}(L,t) \le f(L,t) \le f_{ub}(L,t)  $$
$$ g_{lb}(L,t,1) \le g(L,t,1) \le g_{ub}(L,t,1)$$
	
	\begin{claim}
		\label{int_lb}
		If function $h(x)$ is non-negative and monotonically decreasing in interval $I=[\alpha,\beta]$, and $a<b\in I\cap\mathbb{Z}$ then $\sum_{k=a}^{b}h(k) \ge \int_{a}^{b}h(x)\dx$.
	\end{claim}
	\begin{proof}
		$\sum_{k=a}^{b}h(k) = \sum\limits_{k=a}^b \int_k^{k+1} h(k)\dx \geq
		\sum_{k=a}^{b-1}\int_{k}^{k+1}h(k)\dx \ge \sum_{k=a}^{b-1}\int_{k}^{k+1}h(x)\dx = \int_{a}^{b}h(x)\dx$
	\end{proof}
	\begin{claim}
		\label{int_ub}
		For the same conditions, $\sum_{k=a}^{b}h(k) \le \int_{a}^{b+1}h(x-1)\dx$.
	\end{claim}
	\begin{proof}
		$\sum_{k=a}^{b}h(k) = \sum_{k=a}^{b}\int_{k}^{k+1}h(k)\dx \le
		\sum_{k=a}^{b}\int_{k}^{k+1}h(x-1)\dx = \int_{a}^{b+1}h(x-1)\dx$
	\end{proof}
	
	In \cite{bestharmonic} the authors define the following upper and lower bounds to the Harmonic Sequence:
	\begin{equation}
	\label{harmonic_b}
	\frac{1}{2n + \frac{1}{1-\gamma} -2}  \le H_n -\ln{n} -\gamma < \frac{1}{2n + \frac{1}{3}}
	\end{equation}
	
	\begin{itemize}
		\item $f_{lb}(L,t)$: 
		$$f(L,t) = \sum\limits_{x=t}^L \frac{1}{x} = H_L - H_{t-1} > \underbrace{\frac{1}{2L + \frac{1}{1-\gamma} -2} +\ln{L} +\gamma}_\text{lower-bound} - \underbrace{\left(\frac{1}{2(t-1) + \frac{1}{3}} + \ln{(t-1)} +\gamma\right)}_\text{upper bound} $$
		$$ = \ln{(\frac{L}{t-1})} + \frac{2(t-1) -2L - \frac{1}{1-\gamma} + \frac{7}{3} }{(2L+\frac{1}{1-\gamma}-2)(2(t-1)+\frac{1}{3})} \equiv f_{lb}(L,t)$$
		
		\item $f_{ub}(L,t)$:
		$$ f(L,t) = \sum\limits_{x=t}^L \frac{1}{x} = H_L - H_{t-1} < \underbrace{\frac{1}{2L + \frac{1}{3}} + \ln{L} +\gamma}_\text{upper-bound} - \underbrace{\left(\frac{1}{2(t-1) + \frac{1}{1-\gamma} -2} +\ln{(t-1)} +\gamma\right)}_\text{lower bound} $$
		$$ = \ln{(\frac{L}{t-1})} + \frac{2(t-1) -2L + \frac{1}{1-\gamma} - \frac{7}{3} }{(2L+\frac{1}{3})(2(t-1)+\frac{1}{1-\gamma}-2)} \equiv f_{ub}(L,t)$$
		
		\item $g_{lb}(L,t,1)$:
		Using Claim~\ref{int_lb}, $g(L,t,1) \ge \frac{2}{L-1} \int_{t}^{L} \frac{L-x}{x+1}\dx = ... = \frac{2}{L-1}\left((L+1)\ln(\frac{L+1}{t+1})+(t-L)\right)$
		
		\item $g_{ub}(L,t,1)$:
		Using Claim~\ref{int_ub}, $g(L,t,1) \le \frac{2}{L-1} \int_{t}^{L} \frac{L-x+1}{x}\dx = ... = 2\cdot\frac{L+1}{L-1}\ln(\frac{L+1}{t}) -2\cdot\frac{L+1-t}{L-1}$
	\end{itemize}

\subsection{Lower Bound}
Assume $L$ is sufficiently large, and $t=0.2\cdot L$.
$$f_{ub}(L,0.2 L) = \ln\left(\frac{L}{0.2 L-1}\right) + \frac{-1.6 L + \frac{1}{1-\gamma} - \frac{13}{3} }{(2L+\frac{1}{3})(0.4 L +\frac{1}{1-\gamma}-4)} \approx \ln(5) - \underbrace{\frac{1.6L +O(1)}{0.8L^2 + O(L)}}_{\approx0} \approxeq 1.609$$
$$g_{lb}(L,0.2 L,1) = \frac{2}{L-1}\left((L+1)\ln\left(\frac{L+1}{0.2 L+1}\right)-0.8 L\right) \approx 2\ln(5) - 1.6 \approxeq 1.619$$

$$f(L,t) < f_{ub}(L,0.2 L) \approx 1.609 < 1.619 \approx g_{lb}(L,0.2 L, 1) < g(L,t,1)$$

Since $f(L,t) < g(L,t,1)$ then the rational agent will cheat.

\subsection{Upper Bound}
Assume $L$ is sufficiently large, and $t=0.21\cdot L$.
$$g_{ub}(L,0.21 L, t) = 2\cdot\frac{L+1}{L-1}\ln\left(\frac{L+1}{0.21 L}\right) -2\cdot\frac{0.79 L + 1}{L-1} \approx 2\ln(0.21^{-1}) -2\cdot0.79 \approxeq 1.541$$
$$f_{lb}(L,0.21 L) = \ln{\left(\frac{L}{0.21 L-1}\right)} + \frac{-1.58 L - \frac{1}{1-\gamma} + \frac{1}{3} }{(2L+\frac{1}{1-\gamma}-2)(0.42 L-\frac{5}{3})} \approx \ln(0.21^{-1}) - \underbrace{\frac{1.58 L + O(1)}{0.84 L^2 + O(L)}}_{\approx 0} \approxeq 1.560$$

$$g(L,t,1) < g_{ub}(L,0.2 L,1) \approx 1.541 < 1.560 \approx f_{lb}(L,0.2 L) < f(L,t)$$

Since $f(L,t) > g(L,t,1)$ the rational agent will not cheat.
\end{proof}

\section{Knowledge Sharing}
\label{section_ks}

Here we apply the method from Section~\ref{section_basis} to the problem of Knowledge Sharing
\emph{in a ring},
and show that it satisfies the Linear Threshold enhancement (Section~\ref{enhance1}),
thus the minimal threshold $t$ can be found in $O(L^2 \log L)$ time.

Recall that $e_0(x)$ denotes the expected utility of an honest agent at the beginning of the algorithm,
assuming there are $x$ agents in the network, and $e_m(x)$ denotes the expected utility
of an agent that deviates from the algorithm and duplicates itself $m$ times \emph{successfully}.
For Knowledge Sharing with $k$ possible outputs,
due to the Full Knowledge property (Definition~\ref{def_fullknow})
it holds that $\forall{x}: e_0(x)=\frac{1}{k}$.

In \cite{gtcolor}, the authors show a Knowledge Sharing algorithm for ring graphs
such that when a cheating agent does not pretend to be more than $n$ agents,
i.e., $m < n$, the algorithm is in equilibrium.
Thus, for these cases it holds that $e_m(n)=e_0(n)$,
since the agent duplicates successfully (according to the definition of $e_m(x)$)
but is still unable to affect the Knowledge Sharing algorithm outcome.
In cases where $m \geq n$, the deviating agent controls the output of the Knowledge Sharing algorithm,
and thus its expected utility is $1$, i.e., $e_m(n)=1$.
More formally:
$$ e_m(x)=
\begin{cases}
	0 & m > L-x \\
	1 & x \leq m \leq L-x \\
	\frac{1}{k} & m < x
\end{cases} $$

Notice that this necessarily implies that $m \geq t$.	
In addition, we have the probability of successfully duplicating $m$ times,
which is $p_m(x) = \frac{\binom{L-x}{m}}{\binom{L-1}{m}}$.
Following Equation~\ref{eq_basis} from our solution basis (Section~\ref{section_basis}),
the algorithm is in equilibrium when the following holds:
$$ \sum\limits_{x=t}^L \frac{1}{k} \geq \max_m \sum\limits_{x=t}^{L-m} p_m(x) e_m(x) $$

We can now expand this expression:
\begin{equation*}
\sum\limits_{x=t}^L \frac{1}{k} \geq \max_m \left(
	\sum\limits_{x=t}^m p_m(x) +
	\sum\limits_{x=m+1}^{L-m} \frac{1}{k} \cdot p_m(x)
\right)
\end{equation*}
\begin{equation}
\label{eq_ks}
\sum\limits_{x=t}^L \frac{1}{k} \geq \max_m \left(
	\sum\limits_{x=t}^m \frac{\binom{L-x}{m}}{\binom{L-1}{m}} +
	\sum\limits_{x=m+1}^{L-m} \frac{1}{k} \cdot \frac{\binom{L-x}{m}}{\binom{L-1}{m}}
\right)
\end{equation}

From this result it is clear that for a large enough $k$, the algorithm is not in equilibrium.
In addition, the right side is monotonously decreasing with $t$,
and clearly there is a minimal one for which the algorithm is in equilibrium.

\begin{theorem}
	Knowledge Sharing satisfies the requirements for Linear Threshold (Section~\ref{sub_enhance}).
\end{theorem}
\begin{proof}
To show that Knowledge Sharing satisfies the Linear Threshold enhancement,
note that it holds that $\forall{L,t}: f(L,t)=\frac{1}{k}$.
Additionally, for $t=L$ it holds that $n=t=L$,
thus the number of agents in the network is a-priori known to all agents.
By our assumption that the algorithm is in equilibrium when $n$ is known: $\forall{L,m>0}: g(L,L,m)=0$.

It holds that $\forall{L,t,m'}: g(L,t+1,m') \leq g(L,t,m')$,
since increasing $t$ only removes items from the summation.
By definition it also holds that $\forall{L,t,m'}:~g(L,t,m') \leq \max_m g(L,t,m)$,
thus clearly $\max_m g(L,t+1,m) \leq \max_m g(L,t,m)$,
i.e., $g$ is decreasing as $t$ increases.

Fix $L$, we will separate into two cases:
\begin{enumerate}
	\item $f(L,3) \geq \max_m g(L,3,m)$: Set $t_0=3$.
	Since $g(L,t,m)$ is decreasing with $t$, and $f(L,t)$ is constant ($\frac{1}{k}$),
	then $t_0=3$ satisfies Equation~\ref{eq_linear}.
	
	\item $f(L,3) < \max_m g(L,3,m)$:
	By the intermediate value theorem there exists a minimal value $t_0 > 3$
	such that $f(L,t_0-1) < \max_m g(L,t_0-1,m)$ and $f(L,t_0) \geq \max_m g(L,t_0,m)$.
	Since $g(L,t,m)$ is decreasing with $t$, and $f(L,t)$ is constant ($\frac{1}{k}$),
	then this $t_0$ satisfies Equation~\ref{eq_linear}.
	
\end{enumerate}

In both cases,
the requirements for the Linear Threshold enhancements are satisfied.
\end{proof}

In contrast to the Leader Election problem,
it is not sufficient to check the case where $m=1$ since the right side might have more then one local maximum.
We could not find a different duplication value $m'$ that satisfies this requirement.
However, note that if $k$ is large enough (in regards to $L$)
then the second sum on the right is negligible and the right side does have one maximum.

So, for a general solution which includes any fixed $k$,
we can use binary search to find the minimal threshold $t$,
where for each value of $t$ the $\argmax$ of $m$ is found by "brute force",
leading to a running time of $O(L^2 \log L)$.

\section{Conclusions and Future Work}
\label{section_discussion}

Large networks with an unknown number of participating nodes are common
in realistic scenarios of distributed algorithms.
In this paper, we extended the bounded knowledge model of \cite{gtcolor}
to examine the effects of a limited $id$ space 
on the equilibrium of distributed algorithms.

When the exact size of the network is not known,
an agent can duplicate itself to affect the algorithm outcome,
and the amount of duplications is limited only by its risk of being caught,
determined according to the size of the $id$ space and the threshold $t$,
compared to its profit from duplication, determined according to the algorithm
and the number $m$ of duplicated agents it creates.

Our method can be used to find the \emph{minimal} threshold $t$
such that agents have no incentive to cheat.
That is, we can use the method to find the minimal value $t$ that satisfies $t \leq n \leq L$,
such that the algorithm is still in equilibrium.
Furthermore, we have demonstrated the usage of this method on two problems:
Leader Election and Knowledge Sharing.

Our results provide a number of interesting directions for future research:

\begin{enumerate}

\item One important assumption which was made is that $n\sim U[t,L]$. In real cases, this might not be true. For instance, $n$ can be assumed to be distributed normally or exponentially.
Therefore, the minimal threshold $t$ can be found for more distributions.

\item For simplicity reasons, we assumed the cheater has to decide in advance
on the number of duplications $m$ and on all the duplication $id$s.
However, the cheater can act "smarter" and choose $m$ and $ids$ sequentially, as the algorithm executes.
This setting may increase its profit and thus may increase the minimal threshold $t$ to larger values,
and should be examined separately. 

\item For Knowledge Sharing, an optimal value for $m$ (as a function of $t,L$) might exist,
thus decreasing the running time for finding the minimal threshold $t$.

\item Further fundamental distributed computing problems can be research under our model,
such as Consensus or Spanning Tree.

\item In \cite{gtcolor},
the agents are given lower and upper bounds $\alpha,\beta$ on $n$ such that $\alpha \leq n \leq \beta$.
Future research can explore the combination of both models
such that $\alpha=t$, $\beta \leq L$ (instead of $\beta=L$ as it is currently).

\end{enumerate}

\section{Acknowledgment}
We would like to thank Yehuda Afek for helpful discussions
and his course on Distributed Computing which has inspired this research,
and to Sivan Schick for his contributions to this paper.

\clearpage
\bibliographystyle{abbrv}

\end{document}